\documentclass[11pt]{article}
\usepackage[T1]{fontenc}
\usepackage[utf8]{inputenc}
\usepackage{planargraph}
\usepackage{comment}
\usepackage{amsthm}
\usepackage{graphicx}
\usepackage{subcaption} % for side-by-side figures
\usepackage{url}
\usepackage{hyperref}
\usepackage{enumitem}
\newtheorem{invariant}{Invariant}
\newtheorem{lemma}{Lemma}
\newtheorem{theorem}{Theorem}
\newtheorem{corollary}{Corollary}
\newenvironment{claim}{\textbf{Claim: }\em}{}
\newsavebox{\ftabbingbox}
\newenvironment{ftabbing}
  {\begin{lrbox}{\ftabbingbox}
   \begin{minipage}{\dimexpr\linewidth-2\fboxsep-2\fboxrule}
   \begin{tabbing}}
  {\end{tabbing}
   \end{minipage}
   \end{lrbox}
   \fbox{\usebox{\ftabbingbox}}}

\begin{document}
\title{Simple analysis of an algorithm for multiple-source shortest paths in
  planar graphs}
\author{Philip N. Klein\\Brown University}
\date{}
\maketitle
\begin{abstract}
This paper addresses the following problem: given a planar embedding
graph, compute a representation of the shortest-path trees rooted at
all the boundary nodes of the graph.  Klein gave an $O(n \log n)$
algorithm for this problem; the algorithm subsequently became an
essential ingredient in dozens of algorithms, addressing problems
ranging from distance oracles to edit distance.  However, the
correctness and analysis in that original paper is
complicated and messy and hard to understand.  In this paper, we give
a simple and clear analysis.
\end{abstract}  

\newpage

%\note{Notation $T(a,b]$}

\section{Introduction}

\subsection{The problem}
\label{sec:problem}

In SODA 2005~\cite{Klein05}, Klein presented a problem and an algorithm for the
problem.  At a high level, the problem is as follows:
\begin{itemize}
  \item \emph{input:} A strongly connected directed planar embedded graph with
    arc lengths
  \item \emph{output:} for each node on the boundary of the infinite
    face, a shortest-path tree rooted at that node.
  \end{itemize}
We will use the term \emph{boundary nodes} to refer to the
nodes on the boundary of the infinite face.

In the worst case, an $n$-node planar graph could have $\Omega(n)$
boundary nodes.  A direct
representation of a shortest-path tree has size $\Omega(n)$, so
in the worst case the output would have size $\Omega(n^2)$.

However, there is a more compact representation of these trees: using
\emph{diffs} (a.k.a. deltas).  Let $r_0\ r_1\ \cdots\ r_k$ be the sequence
 of boundary nodes as they appear in counterclockwise order along the boundary of the
infinite face.\footnote{A node can occur more than once.}
For $j=0,1\ldots, n$, let $T_j$ be a $r_j$-rooted shortest-path
tree.  The compact representation consists of the $r_0$-rooted
shortest-path tree and, for $j=1, \ldots, k$, the symmetric difference
in arcs between $T_{j-1}$ and $T_j$.

Why should these compact representation be any smaller than
$\Omega(n^2)$ in the worst case?  Assume for now that for any two
nodes $s$ and $t$, there is a unique shortest $s$-to-$t$
path.\footnote{This assumption is not used by the analysis in this
paper, just in this quick illustrative proof.}
Let $G$ be the input graph, and let $\overleftarrow{G}$ be the graph
obtained by reversing the orientation of each arc.  For each node
$v$ in $\overleftarrow{G}$, let $\overleftarrow{T_v}$ denote the
$v$-rooted shortest-path tree. Planarity implies the following
property:
\begin{quote}
For each arc $vw$ whose tail is $v$, the boundary nodes in
  the $w$-rooted subtree of $\overleftarrow{T_v}$ form a consecutive
  subsequence of the permutation cycle $(r_0\ r_1\ \cdots\ r_k)$.
\end{quote}
This implies a property for $G$:
\begin{quote}
For each arc $wv$ of $G$, the boundary nodes $r_j$ for which $T_j$
contains $wv$ form a consecutive subsequence of $(r_0\ r_1\ \cdots\ r_k)$.
\end{quote}
This property implies that, for each arc $wv$, there is at most one boundary
node $r_j$ such that $wv$ is in $T_{j-1}$ and not in $T_j$, and at
most one boundary node $r_j$ such that $wv$ is in $T_j$ and not in
$T_{j-1}$.  Thus the sum of sizes of the symmetric differences is
twice the number of arcs.

The problem of computing the diffs was called~\cite{Klein05}
\emph{multiple-source shortest paths (MSSP) 
  in planar graphs}, and gave an $O(n \log n)$ algorithm for it.
The algorithm is based on using a link-cut tree to represent a
spanning tree of the dual planar graph, namely the spanning tree
consisting of edges \emph{not} in the primal shortest-path tree.  The
algorithm iterates over the boundary nodes in order, performing
insertions and deletions to transform the $r_0$-rooted shortest-path
tree to the $r_1$-rooted tree, and thence to the $r_2$-rooted tree and
so on.  (We refer to such an arc insertion+deletion as a \emph{pivot}.)

The algorithm is quite simple but the analysis in that paper is quite complicated
and baroque and messy.  The purpose of this paper is to give a simple and clear
analysis.

The author showed that the output could be used to address other problems.
Given a set $S$ of pairs $(v_i, v)$ of nodes in which
the first node is a boundary node, there is an $O((|S|+n) \log n)$
algorithm to compute the $v_i$-to-$v$ distance for every pair $(v_i,
v)$ in $S$.  We call this \emph{batch processing}.
For example, if the number of boundary nodes is
$O(\sqrt{n})$ then all boundary-node-to-boundary-node distances
can be computed in $O(n \log n)$.

%optional
The author mentioned that, applying a partial-persistence
technique~\cite{DSST} to a link-cut tree representing the \emph{primal}
shortest-path tree, one obtains a data structure of size $O(n \log n)$
that supports origin-to-destination distance queries when the origin
is a boundary node.

\subsection{Other algorithms for MSSP}

A similar $O(n \log n)$ algorithm was subsequently proposed by Cabello and Chambers
and Erickson~\cite{CabelloC07,CabelloChambersErickson}.
Their algorithm runs in $O(g n \log n)$ time for graphs of genus~$g$.
 The Cabello et al. algorithm uses the
same structure as the Klein algorithm but uses a different rule for
selecting the next arc to insert into the shortest-path tree.
This rule makes the analysis much simpler.  However, in contrast to
the Klein algorithm, 
it requires the assumption that shortest paths are unique; the authors suggest
overcoming that assumption by use of randomization or by a
lexicographic tie-breaking rule.  (The latter increases the running
time by an $O(\log n)$ factor.)

Eisenstat and Klein~\cite{eisenstat2013linear} gave a simple
tie-breaking rule that does not increase the running time of Cabello
and Chambers and that removes the requirement of unique shortest paths;
however, the
argument is somewhat complicated.  They also
showed that the algorithm of Cabello et al. could be adapted to run in
$O(n)$ time when the costs are integers whose sum is $O(n)$.

Das, Kipouridis, Probst Gutenberg and Wulff-Nilsen~\cite{DKGW22}
gave an MSSP algorithm using an approach very different from that used by Klein
and by Cabello et al.  The compact representation of shortest-path
trees obtained is very different; it does not represent the diffs.
Their algorithm takes $O(n \log k)$ time where $k$ is the number of
boundary nodes.  It computes a data structure that supports
origin-to-destination queries in $O(\log k)$ time where the origin is
a boundary node.\footnote{
In work that is under review~\cite{KleinDragos}, these three algorithms were evaluated
experimentally on planar graphs derived from road networks.  In our
experiments, the
algorithm of Das et al. was an order of magnitude slower than the
other two.  The algorithm of Cabello et al. was slightly faster than
the algorithm of Klein; the primary reason is that sometimes Klein's
algorithm pivots into the tree an arc that is not in any shortest-path
tree rooted at a boundary node (and pivots it out later, before moving
to the next boundary node).}

\subsection{Uses of MSSP}

The MSSP algorithm has been used or adapted for many applications involving
 planar or bounded-genus graphs:
 in computing graph diameter~\cite{cabello2018subquadratic,li2019planar,chan2019faster},
 in computing many origin-to-destination distances~\cite{cabello2012many},
 in computing shortest paths in the presence of negative-cost arcs~\cite{KleinMozesWeimann2010,MozesWulffNilsen2010,karczmarz2024max},
in exact distance oracles~\cite{cohen2017fast,nussbaum2011improved,mozes2012exact,long2021planar,gawrychowski2018better,charalampopoulos2019almost,Boneh0MPW25,mozes2026distancesplanargraphsfree} and approximate distance oracles~\cite{Klein05,kawarabayashi2011linear,chan2019faster},
in computing distances with failing arcs~\cite{wulff2010solving,emek2010near,chechik2020simplifying,erickson2011computing} or vertices~\cite{charalampopoulos2022exact},
in computing a subgraph that exactly or approximately preserves distances~\cite{chang2018near},
in computing a minimum $s,t$ cut and a maximum $s,t$ flow in undirected planar graphs~\cite{italiano2011improved},
in global minimum cut in undirected~\cite{lkacki2011min,italiano2011improved} and directed~\cite{mozes2018minimum} planar graphs and in bounded-genus graphs~\cite{chambers2023minimum},
in computing an oracle for minimum $s,t$ cut~\cite{borradaile2015min,borradaile2016all,karczmarz2024max}
in computing girth~\cite{chang2013computing}, in edit distance~\cite{cassis2023optimal,charalampopoulos2025pattern},
in computing single-source all-sinks max flow~\cite{lacki2012single},
in maximum flow with multiple sources and sinks~\cite{borradaile2017multiple},
in combinatorial topology~\cite{de2010tightening,erickson2011shortest,erickson2011minimum,fox2013shortest,borradaile2017minimum},
in distance problems in continuous planar graphs~\cite{cabello2025algorithms},
in minimum-cost bipartite matching in planar graphs~\cite{asathulla2019faster,lahn2019faster},
in shortest vertex-disjoint paths in undirected planar graphs~\cite{borradaile2015towards},
and in replacement
paths~\cite{wulff2010solving,emek2010near,erickson2011computing,chechik2020simplifying}.

For all these applications, batch processing suffices and 
the Cabello et al. algorithm or the Das
et al. algorithm can be used
in place of the Klein algorithm. However, there is one important
application for which the Klein algorithm seems to be needed:
approximation schemes for Steiner-type problems in planar graphs.
Klein~\cite{Klein06} gave a polynomial-time approximation scheme
(PTAS) for Steiner Traveling Salesman (a.k.a. subset TSP), in which
one is given a subset of nodes and one seeks a tour that visits at
least those nodes.  That paper introduced a construction that was also
used in a paper by Borradaile, Klein, and
Mathieu~\cite{BorradaileKleinMathieu2009} giving a PTAS for Steiner
tree in planar graphs; the latter paper called this construction
\emph{strip decomposition}.

The naive algorithm for strip decomposition takes $\Theta(n^2)$
time.  However, Klein~\cite{Klein06} sketched an algorithm that adapts
the original MSSP algorithm to achieve $O(n \log n)$ time.  This
enables the approximation schemes for Steiner TSP and Steiner tree to
run in $O(n \log n)$ time instead of $O(n^2)$ time.  The algorithm
makes essential use of the pivot rule of~\cite{Klein05}.

The strip decomposition (the most complex step in what was called
\emph{brick decomposition}) came to be essential also to
near-linear realization of
approximation schemes for generalizations of Steiner tree: Steiner
two-edge connectivity~\cite{borradaile2008two},
prize-collecting Steiner tree and TSP~\cite{BateniCEHKM11},
and Steiner forest~\cite{EfficientSteinerForestPTAS2012}
(see~\cite{BateniHM11}).

Because of its role in strip decomposition and brick decomposition,
the original MSSP algorithm remains important.  However, as mentioned,
the original analysis is awful.  This paper gives a clear and
intuitive analysis of that algorithm.

\section{Background and terminology}

\subsection{Graphs and darts and edges}

Consider a directed graph $G$.  We posit for notational convenience
(as in some max-flow algorithms) that corresponding to each arc $uv$
there is an arc $vu$ in the opposite direction.  We refer to this as
the \emph{bidirectionality assumption}.
The two arcs can have different lengths, and in fact
one can have length infinity, so this assumption is without loss of
generality. 
We use $\rev(\cdot)$ to denote the function that maps
each arc to the oppositely oriented arc.
An arc is considered oriented from its tail to its head.
We denote the head and tail of an
arc or dart $a$ by $\head_G(uv)$ and $\tail_G(uv)$ respectively.

We use the term \emph{edge} when we do not want to distinguish between
the two orientations $uv$ and $vu$.\footnote{As the graph theorist
Berge wrote, ``It would be convenient to say that there are two
theories and two kinds of graphs: directed and undirected.  This is
not true.  {\em All graphs are directed,} but sometimes the direction
need not be specified.''}

\subsection{Rooted trees and paths}

For a rooted spanning tree $T$ in a graph.  We say $T$ is $r$-rooted
if $r$ is the root of $T$.  For a node $v$, we denote by $T[v]$ the
root-to-$v$ path in $T$.  For nodes $u$ and $v$, we denote by $T[u,v]$
the $u$-to-$v$ path in $T$.

The tree can be interpreted as a set of edges when we do not care
about orientation.  When we want to think of the tree as consisting of
paths from the root towards the leaves, we refer to the tree as
\emph{leafward}.  When we want to think of the tree as consisting of
paths from the leaves toward to the root, we use the term \emph{rootward}.

Let $uv$ be an arc such that the edge of $uv$ is not represented in $T$.
The \emph{fundamental cycle of $uv$}
is the simple cycle that consists of the arc $uv$ and the path $T[v,u]$.

For a path $P$ and nodes $u$ and $v$ such that $u$ precedes
$v$, we denote by $P[u,v]$ the $u$-to-$v$ subpath of $P$.  We denote
by $\pstart(P)$ and $\pend(P)$ the start and end of $P$.

A path is \emph{trivial} if its start and end are the same and it
consists of zero arcs.  For paths $P$ and $Q$ such that
$\pend(P)=\pstart(Q)$, we denote by $P\circ Q$ the path that first
traverses $P$ and then traverses $Q$, i.e. the concatenation of the
paths $P$ and $Q$.

\subsection{Lengths and shortest paths and slack costs}
\label{shortest-path-subsec}

For an assignment $\ell(\cdot)$ of \emph{lengths} (also called
\emph{costs}) to edges of the graph and a path or cycle $P$, we denote
by $\ell(P)$ the length of $P$, i.e. the sum of values $\ell(\cdot)$
assigns to the edges of $P$.

% todo: SHOULD THIS BE INCLUDED?
%% We assume no
%% cycle has negative length.  It follows from this assumption that removing cycles from a nonsimple
%% path does not increase its length.  Hence if there is a shortest
%% $u$-to-$v$ path, there is one that is simple.

For any labeling $\phi(\cdot)$ of the nodes of a graph with
numbers, there is a new length assignment $\ell_\phi$ defined by
\begin{equation}
\label{reduced-length}
\ell_\phi(uv) = \phi(u) +\ell(uv) - \phi(v)
\end{equation}

 We refer to the lengths $\ell_\phi(uv)$ as {\em slack costs with
  respect to $\phi$}.  For any nodes $s$ and $t$,  for any $s$-to-$t$
  path $P$,
\begin{equation}\label{path-reduced-length}
\ell_\phi(P)
= \ell(P) + d(s) - d(t)
\end{equation}
It follows that an $s$-to-$t$
path is shortest (i.e. is a minimum-cost $s$-to-$t$ path) according to
$\ell_\phi$ if and only if it is shortest according to $\ell$.

An arc $uv$ is a {\em relaxed arc} with respect to $\phi$ if its
reduced length is nonnegative, i.e. if
\begin{equation} \label{relaxed-eqn}
\phi(u) + \ell(uv) \geq \phi(v) 
\end{equation}
and is a {\em tense arc} otherwise.
% todo: retain this?
%%   An arc is {\em tight}
%% with respect to $\phi$
%% if its reduced length is zero, i.e. if~(\ref{relaxed-eqn}) is
%% satisfied with equality.

Suppose that every arc is relaxed and, for some nodes $r$ and $v$,
$\phi(r)=0$ and $\phi(v)$ is the length of {\em some} $r$-to-$v$ path $P_v$.
It follows by~(\ref{path-reduced-length}) that $\ell_\phi(P_v) = 0$.
Every arc and hence every $r$-to-$v$ path has nonnegative reduced
length, so $P_v$ is a shortest path with respect to 
the reduced lengths $\ell_\phi$, and hence also with respect to the
original lengths $\ell$.

Let $T$ be an $r$-rooted leafward spanning tree.  We define the
labeling $\phi_T(\cdot)$ by $\phi_T(v) = \ell(T[v])$.  We then define
the slack costs with respect to $T$ as the slack costs with respect to 
$\phi_T(\cdot)$, and we accordingly define an arc to be relaxed or
tense with respect to $T$.

\subsection{Planar embedded graphs}

It is convenient to work with \emph{combinatorial} embeddings.
Under
the bidirectionality assumption and the assumption that there are no
isolated nodes, we define an embedding of a graph $G=(V, A)$ to be a
permutation $\pi: A \longrightarrow A$ whose set of orbits is
$\set{\set{\text{arcs with head} v}\ : \ v\in V}$.
For each arc $uv$, $\pi(uv)$ is another arc $wv$ with the same head $v$;
in particular, it is the \emph{next} arc as one encircles the node $v$
in the counterclockwise direction.  This definition is useful when
defining or describing properties related to embeddings, e.g. see
Figure~\ref{fig:interval}.
Using combinatorial embeddings, one can formally define faces,
the planar dual, what it means for a cycle to enclose a face, edge, or
node.  See~\cite{planarity} for formal definitions and results.

One such fundamental result:
\begin{quote}
  If $T$ is a spanning tree of planar embedded graph $G$,
  the set of edges of the planar dual that are not in $T$ form a
  spanning tree $T^*$ of the dual $G^*$.
\end{quote}
This result is central to the MSSP algorithm.

\section{The algorithm}

Let the nodes of the boundary, in counterclockwise order about the
graph, be $r_0 r_1 \ldots r_k$.  The algorithm maintains a tree $T$
rooted at one of the boundary
nodes, and the set $T^*$ of arcs of the graph that do not belong to
$T$ (which is a spanning tree of the dual planar graph).
We interpret $T^*$ as rooted at the infinite face, oriented towards
the root (so a nonroot node is farther in $T^*$ from the root than its
parent).

%% Initially $T$ is an $r_0$-rooted shortest-path
%% tree.\footnote{In the original algorithm description of~\cite{Klein05}, it was
%% required to be the \emph{rightmost} shortest-path tree but the new
%% analysis does not require this.}

\noindent At a high level, the algorithm is as follows:
\medskip

\noindent \begin{ftabbing}
initialize $T := $ $r_0$-rooted shortest-path tree\\
for $j=0,1, \ldots, k-1$:\\
\quad \= (\textbf{special pivot}:) eject from $T$ the arc whose head
is $r_{j+1}$, and inject the arc $r_{j+1} r_j$\\
      \> while there is a tense arc with respect to $T$:\\
      \> \quad \= (\textbf{leafmost rule}:) let $a$ be a tense arc,
      chosen to be leafmost in $T^*$ among tense arcs.\\
      \>       \> (\textbf{ordinary pivot}:) eject from $T$ the arc
        whose head is the head of $a$ in $G$, and inject $a$.
\end{ftabbing}
\medskip

\noindent The algorithm performs two kinds of pivots.  An \emph{ordinary} pivot,
given an arc $uv$ not in $T$, removes the arc in $T$ whose head is
$v$, and inserts $uv$.  (This requires that $u$ not be a descendant of
$v$ in $T$ before the pivot.)  A \emph{special} pivot transforms a tree
$T$ rooted at one boundary node $r$ into a tree rooted at the next
boundary node $r'$ in counterclockwise order around the graph.  It consists
of removing the arc of $T$ whose head is $r'$ and adding the arc whose
tail is $r'$ and whose head is $r$.

When the while-loop terminates, because there are no arcs that are tense with
respect to the current tree $T$, that tree is a shortest-path tree.

The spanning tree $T$ of the input graph is maintained using an array
$\pred[]$ that, for each nonroot node $v$, stores the arc $\pred[v]$
of $T$ whose head is $v$.  The spanning tree $T^*$ of the planar dual
is maintained using a link-cut tree~\cite{ST83,planarity,eisenstat2014dtree}.  As a
consequence, as described in~\cite{Klein05}, each iteration can be
carried out in amortized $O(\log n)$ time.  The analysis consists in
showing that each arc is injected at most once, so the number of
iterations is bounded by the number of arcs.  (Note that some
injections involve arcs that belong to no shortest-path tree.)
Assuming the input graph does not have parallel edges, the number of arcs is $O(n)$ so the
running time is $O(n \log n)$.

How to show that each arc is injected at most once?  Recall the the
illustrative proof
in Section~\ref{sec:problem}, in which we showed that, for each node
$v$, for each arc $wv$ the boundary nodes $r_j$ for which $T_j$
contains $wv$ form a consecutive subsequence of the permutation cycle
of boundary nodes.  The analysis of the algorithm shows essentially that for each
node $v$, as the algorithm progresses, the arcs $wv$ entering $v$ get
``used up'' (or become \emph{defunct}) in counterclockwise order about
$v$, and an arc, once it becomes defunct, does not get injected into
the shortest-path tree $T$.

\section{Analysis of the MSSP algorithm}

\subsection{The set-up}

\begin{theorem} \label{thm:summary}
  Each arc enters $T$ at most once.
\end{theorem}  

In the analysis that proves Theorem~\ref{thm:summary}, in each
iteration some arcs are designated as \emph{defunct}.  Initially no
arc is defunct.  Once an arc becomes defunct, it remains defunct for
the remainder of the algorithm's execution.  We give an invariant and
show that the tree $T$ maintained by the algorithm satisfies the
invariant throughout execution.  The rule for designating arcs as
defunct implies that, when an arc is injected into $T$, it becomes
defunct.  To prove the theorem, we use the invariant to show that no defunct arc is
injected into $T$ by (1) an ordinary pivot
(Lemma~\ref{lem:inv:defunct})) or by a special pivot
(Lemma~\ref{lem:special-pivot-not-defunct}).

\begin{invariant} \label{inv:between}
For any path $P$ starting at the root of $T$ and consisting of arcs
  of $T$ and defunct arcs, $\length(T[\pend(P)]) \leq \length(P)$.
\end{invariant}

For $i=0, 1, \ldots$, we denote by $T_i$ the tree $T$ obtained after
$i$ iterations.  We show by induction on $i$ that $T_i$ satisfies the invariant.
The initial tree $T_0$ trivially satisfies the invariant
because initially no arcs are defunct.
Assume the invariant hold during the first $i$ iterations.  We will show it holds
after $i+1$ iterations.

Presently we describe the rule governing which arcs become defunct in
an iteration.  First, some notation.  For distinct darts $a$ and
$\bar a$ with the same head, as illustrated in
Figure~\ref{fig:interval}, we use $(a,\bar a]$ to denote the set of
darts $\pi^k(a)$ where $k$ is positive and is no more than the smallest
positive integer $\ell$ such that $\pi^\ell(a)=\bar a$.

\begin{figure} \centering
  \includegraphics[scale=.5]{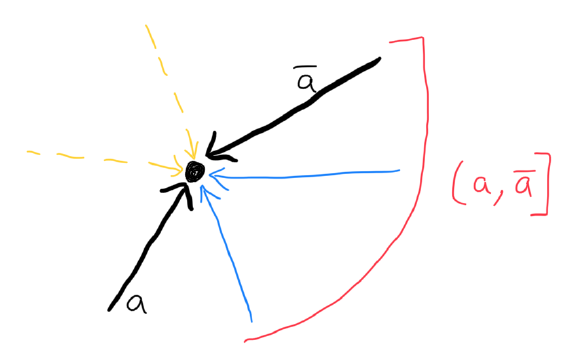}
  \caption{This figure illustrates the definition of $(a,\bar a]$
    where $a$ and $\bar a$ are darts with a common head $v$.  The
    embedding defines a cyclic ordering of darts with head $v$, and
    $(a,\bar a]$ denotes the set consisting of darts that come after
    $a$ and before $\bar a$ (including $\bar a$).}
    \label{fig:interval}
  \end{figure}

\paragraph{Rule for becoming defunct:} Suppose the current tree is $T$
and an ordinary pivot replaces an arc $a$ of $T$ with an arc $\bar a$.  As
illustrated in Figure~\ref{fig:defunct}, an arc $b$ not in $T$ becomes defunct in
this iteration if $T$ contains a (possibly trivial) path $P$ of arcs
with $\pend(P)=\head(b)$ such that the last dart of $b \circ \rev(P)$
belongs to $(a,\bar a]$.  We say that $P$ is the \emph{witness path}.
Note in particular that $\bar a$ becomes defunct after the iteration.

For a special pivot that injects an arc $\bar a$, the only arc that
becomes defunct is $\bar a$.

\begin{figure}
  \centering
  \includegraphics{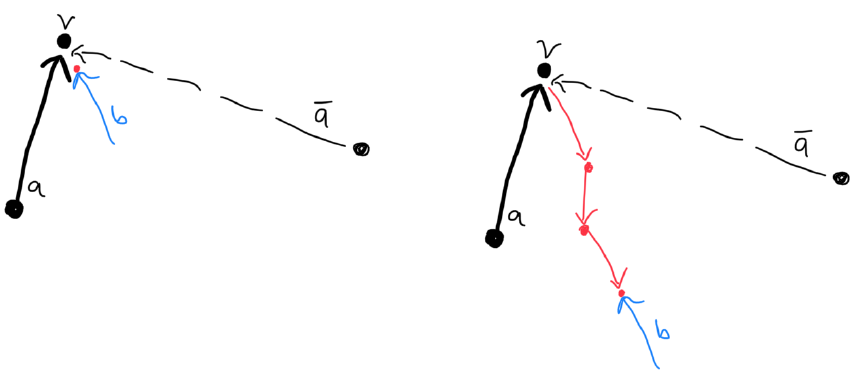}
  \caption{This figure illustrates when an arc $b$ becomes \emph{defunct} during
    an iteration.  The iteration ejects arc $a$ from the tree, and
    injects arc $\bar a$ into the tree.  In the diagram on the left,
    the arc $b$ is such that, in the counterclockwise cycle of arcs
    with head $v$, starting with arc $a$, the arc $b$ comes before the
    arc $\bar a$.  (Formally, $\pi^\ell(a) = b$, $\pi^d(a) = \bar a$, and
    $\ell \leq d$.)  In the diagram on the right, there is a nonempy
    \emph{witness path} of arcs belonging to the tree; the last arc of
    the witness path has the same head as arc $b$, and, starting with
    arc $a$, the reverse of the first arc of the witness path comes
    before the arc $\bar a$ in the counterclockwise cycle of arcs
    with head $v$.}
\label{fig:defunct}
\end{figure}

\subsection{Ordinary pivots}

Suppose $T_{i+1}$ is obtained from
$T_i$ by an ordinary pivot that injects arc $xy$.  By the injection
rule, $xy$ is a leafmost arc that is tense with respect to $T_i$.

\begin{lemma} \label{lem:inv:defunct} $xy$ was not defunct by the
  beginning of iteration $i+1$.
\end{lemma}

\begin{proof} Assume otherwise.  Then, by
  Invariant~\ref{inv:between} applied to $T_i$,
  $\length(T_i[x] \circ xy) \geq \length(T_i[y])$.  However, the
  injection rule implies that $xy$ is tense with respect to $T_i$, so
  $\length(T_i[x] \circ xy) < \length(T_i[y])$.  This is a
  contradiction.
\end{proof}

Let $\bigtriangledown$ be the fundamental cycle of $xy$ with respect
to $T_i$, oriented so that $xy$ is forward.  By leafmostness, every
arc strictly enclosed by $\bigtriangledown$ is relaxed with respect to
$T_i$.

\begin{lemma} \label{lem:auxiliary} There is no suffix of $T_0[y]$ that starts on
  $T_i[x]$ and is enclosed by $\bigtriangledown$ and excludes $xy$.
\end{lemma}

\begin{proof} The proof is illustrated in
  Figure~\ref{fig:no-suffix}. Suppose $P$ were such a suffix.  Since
  $P$ is enclosed by $\bigtriangledown$ and does not include $xy$ and
  since $xy$ is a leafmost tense arc, every arc of $P$ is relaxed
  with respect to $T_i$.  The path $T_i[\pstart(P),x]$ consists of
  arcs of $T_i$, which have zero slack cost with respect to $T_i$.
  The arc $xy$ has negative slack cost with respect to $T_i$.  Thus
  the total slack cost of $T_i[\pstart(P),x]\circ xy$ is negative, so
  this path is shorter than $P$.  This contradicts the fact that $P$
  is part of a shortest-path tree.
\end{proof}

\begin{figure} \centering
  \includegraphics[height=2.5in]{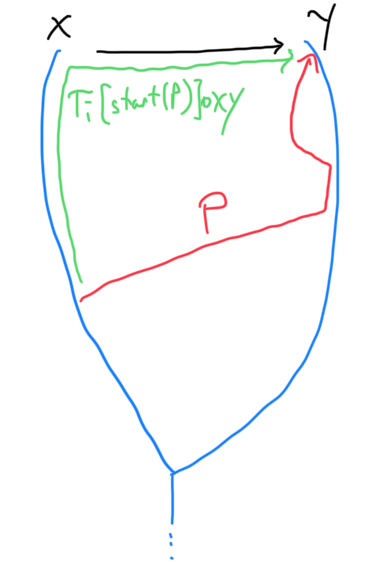}
  \caption{Assume there is a suffix $P$ of $T_0[y]$ that starts on
    $T_i[x]$ and is enclosed by $\bigtriangledown$ and does not
    include $xy$.  The path $T_i[\pstart(P), x]\circ xy$ is
    shorter.}

  \label{fig:no-suffix}
\end{figure}

\begin{lemma} \label{lem:CCW} 
 $\bigtriangledown$ is counterclockwise.
\end{lemma}

\begin{proof} Let $z$ be the leafmost common ancestor of $x$ and $y$
  in $T_i$.
Write $T_i[z,y] = b_1 b_2 \cdots b_s$, and, for $\ell=1, \ldots, s$, let $v_\ell =
\head(b_\ell)$.  (Thus $v_s=y$.)

For each nonnegative integer $k$, define $P_k$ to be the maximal
suffix of $T_k[y]$ that is enclosed by $\bigtriangledown$. Because
$T_k[y]$ starts on the boundary of the graph, it cannot be strictly
interior to $\bigtriangledown$.  It follows that
$\pstart(P_k)$ must be a vertex of $\bigtriangledown$.

Assume for a contradiction that $\bigtriangledown$ is CW.  We prove
the following claim, which implies that, for every $k$,  $P_k\neq T_i[z,y]$,
which contradicts the fact that $P_i=T_i[z,y]$ by definition of $P_i$.

\begin{claim} For each $k$, there is an integer $q(k)\in \set{1, \ldots, s}$
  such that $\pstart(P_k)=v_{q(k)}$.
\end{claim}

The proof is by induction on $k$.  Consider the case $k=0$.
Lemma~\ref{lem:auxiliary} shows that $\pstart(P_0)$ does not belong to
$T_i[z,x]$.  This, together with the fact that $\pstart(P_0)$ must be
a vertex of $\bigtriangledown$, shows that the claim holds for $k=0$.

\begin{figure} \centering
  \includegraphics{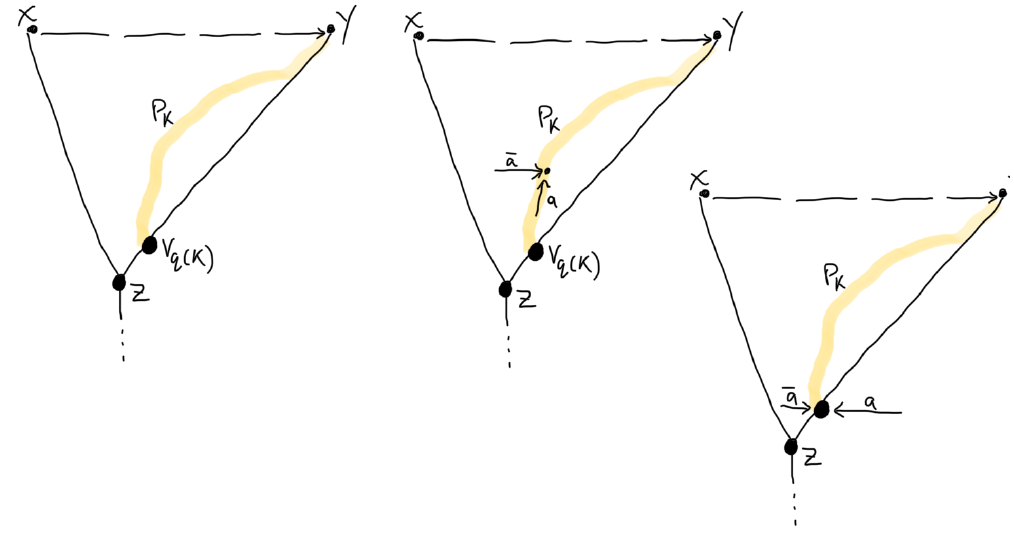}
  \caption{This figure illustrates an argument in the proof of the
    claim.  The left figure shows the cycle $\bigtriangledown$, and $P_k$,
    the maximal suffix of $T_k[y]$ consisting of arcs enclosed by
    $\bigtriangledown$.  In the middle figure, we see the case in
    which iteration $k+1$ injects an arc $\bar a$ that enters $P_k$
    strictly from the left; in this case, the ejected arc $a$ is an
    arc of $P_k$.  Let $v$ be the head of $\bar a$;
    then $P_k[v,y]$ is a witness path showing that the arc $xy$
    becomes defunct in the iteration that injects $\bar a$.
    In the rightmost figure, the same
    argument is illustrated, but where the ejected arc $a$ is the arc
    of $T_k[y]$ immediately preceding $P_k$.}
  \label{fig:defunct-argument}
\end{figure}

\begin{figure} \centering
  \includegraphics[scale=1]{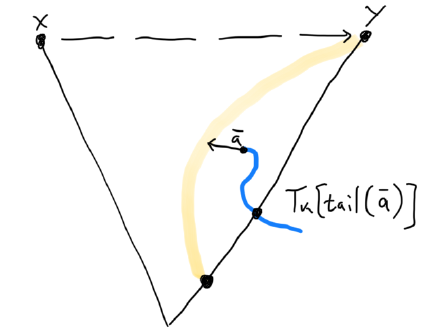}
  \caption{This figure illustrates another argument in the proof of
    the claim.  Iteration $k+1$ injects an arc $\bar a$ that does not
    enter $P_k$ from the left.  In the diagram, $\bar a$ is enclosed
    by $\bigtriangledown$.  In this case, we note that the path
    $T_k[\tail(\bar a)]$ must enter $\bigtriangledown$ at some node in
    $\set{v_{q(k)}, v_{q(k)+1}, \ldots, v_s}$.}
  \label{fig:cannot-cross-Pk}
\end{figure}

Assume the claim holds for $k$.  In particular, $\pstart(P_k)$ is not
the root of $T_k$, so $T_k[y]$ contains some arc $a$ immediately
before $P_k$.  If $\pstart(P_{k+1})=\pstart(P_k)$
then the claim holds for $k+1$, so assume that in iteration $k+1$ an
arc $\bar a$ is injected that results in
$\pstart(P_{k+1})\neq\pstart(P_k)$.  If $\bar a$ enters $a\circ \bar P_k$
from the left at some vertex $v$ then, as illustrated in
Figure~\ref{fig:defunct-argument}, the arc $xy$ would be defunct (with
witness path $P_k[v,y]$), contradicting Lemma~\ref{lem:inv:defunct}.

Suppose therefore that $\bar a$ does not enter $a\circ \bar P_k$
from the left.  If $\bar a$ is
not enclosed by $\bigtriangledown$, the fact that its head \emph{is}
enclosed implies that its head is on the boundary of
$\bigtriangledown$ and that therefore its head is one of $v_{q(k)},
v_{q(k)+2}, \ldots, v_s$, hence the claim holds for
$k+1$.  If $\bar a$ is enclosed by $\bigtriangledown$, it is enclosed
by the cycle $P_k \circ \rev(T_i[v_{q(k)}, y])$.  As illustrated in
Figure~\ref{fig:cannot-cross-Pk}, the path
$T_k[\tail(\bar a)]$ cannot cross $P_k$ because $P_k$ is part of
$T_k$, so that path last enters $\bigtriangledown$ at some vertex in
$\set{v_{q(k)}, v_{q(k)+1}, \ldots, v_s}$, proving the claim holds for
$k+1$.  This completes the induction, proving the claim, which
completes the proof of Lemma~\ref{lem:CCW}.
\end{proof}

\begin{figure} \centering
  \includegraphics[scale=1]{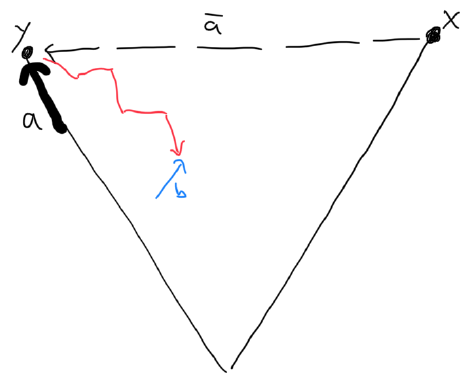}
  \caption{This figure shows that the arcs defunct during iteration
    $i+1$ are enclosed by $\bigtriangledown$.}
  \label{fig:defunct-enclosed}
\end{figure}

\begin{corollary} \label{cor:defunct-enclosed} During iteration $i+1$,
  the arcs that become defunct other than $xy$ are strictly enclosed by $\bigtriangledown$. 
\end{corollary}

\begin{proof} As illustrated in Figure~\ref{fig:defunct-enclosed}, by
  Lemma~\ref{lem:CCW}, any arc $b$ in $(a,\bar a]$ is enclosed by
  $\bigtriangledown$.  In 
  particular, if in iteration $i+1$ arc $b$ is defunct with a trivial
  witness path then $b$ is enclosed by $\bigtriangledown$.
  Furthermore, suppose $P$ is a nontrivial path consisting of arcs of $T_k$ and the
  reverse $b$ of the first arc of $P$ belongs to $(a, \bar a]$.
  Because $T_i$ is a tree, no vertex of $P$ except
  its first is a vertex of the cycle $\bigtriangledown$, so every arc
  of $P$ is strictly enclosed by $\bigtriangledown$ and in particular
  the final vertex of $P$ is strictly enclosed by $\bigtriangledown$.
  This shows that any arc that becomes defunct in this iteration with
  witness path $P$ is enclosed by $\bigtriangledown$.
\end{proof}

\begin{lemma} \label{lem:between-helper}
  Suppose $P$ is a path starting at the root of $T_i$ each of whose
  arcs is in $T_i$ or became defunct during the first $i+1$
  iterations, and $P$ does not include $xy$.  Then $\length(P) \geq
  \length(T_i[\pend(P)])$.
\end{lemma}

\begin{proof} By induction on the size of $P$.  Suppose $P$ has at least one 
  arc, and write $P=P' \circ \bar a$.  By the inductive hypothesis,
  $\length(P') \geq \length(T_i[\pend(P')])$.  Let $Q = T_i[\pend(P')] \circ \bar a$,
  and note that $\pend(Q)=\pend(P)$.  We claim that $\length(Q) \geq \length(T_i[\pend(Q)])$,
 which suffices to prove the lemma.
  If $\bar a$ belongs to $T_i$ then the claim holds trivially, so suppose $\bar
  a$ was spent during the first $i+1$ iterations.  If it was spent
  during the first $i$ iterations then the claim follows directly from
Invariant~\ref{inv:between}.  Suppose $\bar a$ was spent during the
$i+1^{st}$ iteration.  By Corollary~\ref{cor:defunct-enclosed}, $\bar a$ is
enclosed by $\bigtriangledown$ (and is not $xy$).  By the Leafmost
Rule, $\bar a$ is relaxed with respect to $T_i$, whence $\length(Q)
\geq \length(T_i[\pend(Q)])$.
\end{proof}

\begin{lemma} \label{lem:between} 
% If $T$ satisfies Invariant~\ref{inv:between} and
%   a leafmost pivot is performed, the resulting tree satisfies
An ordinary pivot that injects $xy$ preserves 
  Invariant~\ref{inv:between}.
\end{lemma}

\begin{proof}  Assume for a contradiction that $T_{i+1}$ does not
satisfy Invariant~\ref{inv:between}, and let $P$ be a counterexample
that minimizes the number of arcs not in $T_{i+1}$.

Let $P$ be a path starting at the root of $T_{i+1}$ every arc of which
is in $T_{i+1}$ or became defunct after at most $i+1$ iterations.  We need
to show that $\length(P) \geq \length(T_{i+1}[\pend(P)])$.
If $xy$ is not in $P$ then this follows from
Lemma~\ref{lem:between-helper} and the fact that the root of $T_i$ is
the root of $T_{i1}$ and $\length(T_{i+1}[\pend(P)]) \leq \length(T_i[\pend(P)])$.

Suppose $xy$ is in $P$, and write $P
= P_1\circ xy \circ P_2$.  By Lemma~\ref{lem:between-helper}, $\length(P_1)
\geq \length(T_i[\pend(P_1)])$.  Let $Q = T_i[y] \circ P_2$.  By
Lemma~\ref{lem:between-helper}, $\length(Q) \geq
\length(T_i[\pend(Q)])$.

\begin{eqnarray*}
  \length(P) & = & \length(Q) - \length(T_i[y]) + \length(P_1) + \length(xy)\\
             & \geq & \length(Q) -  \length(T_i[y]) + \length(T_i[x]) + \length(xy)\\
             & \geq & \length(T_i[\pend(Q)]) - \length(T_i[y]) + \length(T_i[x]) +\length(xy)\\
             & = & \length(T_{i+1}[\pend(Q)])
\end{eqnarray*}
where the last equality holds because $Q=T_i[y] \circ P_2$ and
$T_{i+1}[\pend(P)] = T_i[x] \circ xy \circ P_2$.
\end{proof}

\begin{figure} \centering
  \includegraphics[scale=.75]{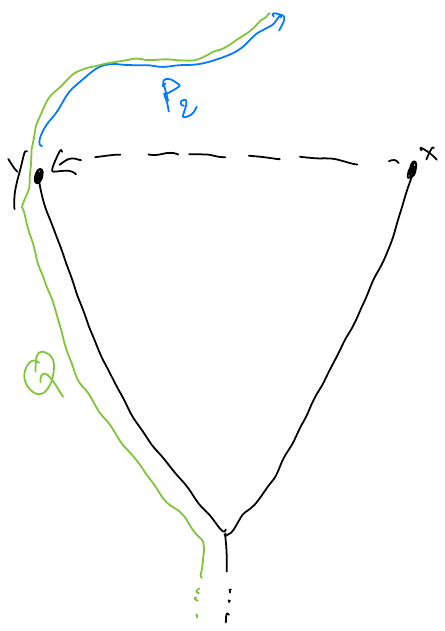}
  \caption{Figure for the proof of Lemma~\ref{lem:between}}
  \label{fig:proof-ordinary-pivot-preserves-invariant}
  \end{figure}

\subsection{Special pivots}

Now we consider the case in which iteration $i+1$ is a special pivot.
Let $r_j$ be the root of $T_i$.  The tree $T_{i+1}$ is obtained from
$T_i$ by removing the arc of $T_i$ whose head is $r_{j+1}$ and adding
the boundary arc $r_{j+1} r_j$.

\begin{lemma} \label{lem:inv-holds-special-pivot}
  If iteration $i+1$ is a special pivot, it preserves Invariant~\ref{inv:between}.
\end{lemma}

\begin{proof}  Assume for a contradiction that
  a path $P$ is a counterexample showing that $T_{i+1}$ does not satisfy the
  invariant.  Then
  \begin{equation}
    \text{every arc of $P$ is in $T_{i+1}$ or is defunct after
      the pivot} \label{eq:arcs-defunct}
      \end{equation}
  \begin{equation}
    \length(P) < \length(T_{i+1}[\pend(P)]) \label{eq:path-short}
    \end{equation}

  \paragraph{Case 1:} $P$ contains $r_{j+1} r_j$.  Write $P = r_{j+1} r_j \circ Q$.
  Note that $\pend(Q)=\pend(P)$.  Then

\begin{equation} \text{every arc of $Q$ is in $T_i$ or is defunct
    before the pivot
 \qquad by (\ref{eq:arcs-defunct})
}
\end{equation}
\begin{eqnarray*}
\length(T_i[\pend(Q)]) & = & \length(T_{i+1}[\pend(Q)]) -
\length(r_{j+1} r_j)\\
& > & \length(P) - \length(r_{j+1} r_j) \qquad \text{by
    (\ref{eq:path-short})}\\
& = & \length(Q)
\end{eqnarray*}

\paragraph{Case 2:} $P$ does not contain $r_{j+1} r_j$.  Let $Q=
T_i[r_{j+1}] \circ P$.  Then
\begin{displaymath}
  \text{every arc of $Q$ is in $T_i$ or is defunct before the pivot \qquad by (\ref{eq:arcs-defunct})}
  \end{displaymath}
\begin{eqnarray*}
  \length(Q) & = & \length(T_i[r_{j+1}]) + \length(P)\\
  & < & \length(T_i[r_{j+1}]) + \length(T_{i+1}[\pend(P)]) \qquad \text{by
    (\ref{eq:path-short})}\\
  & = & \length(T_i[\pend(P)])
  \end{eqnarray*}
  
Thus in either case $Q$ is a counterexample showing that $T_i$ does not satisfy the
  invariant.
\end{proof}

\begin{lemma} \label{lem:special-pivot-not-defunct}
The arc $r_{j+1} r_j$ was not defunct by the beginning of iteration~$i$.
\end{lemma}

\begin{proof}
Assume for a contradiction that $r_{j+1} r_j$ was made defunct in iteration
$i'$ such that $i' \leq i$.
That iteration cannot have been a special pivot because.
this arc had not previously been involved in such an iteration.
The iteration must have been an ordinary pivot.

By Lemma~\ref{cor:defunct-enclosed}, that pivot injected an arc $xy$
such that either $xy=r_{j+1} r_j$ or the fundamental cycle $\bigtriangledown$ of $xy$ with
respect to $T_{i'}$ strictly encloses $r_{j+1} r_j$.  The latter
cannot happen because $r_{j+1} r_j$ is on the boundary;
$xy$ cannot be $r_{j+1} r_j$ because, by Lemma~\ref{lem:CCW},
$\bigtriangledown$ is counterclockwise and $r_{j+1} r_j$ is a
clockwise boundary arc.
\end{proof}

%% \begin{proof} Suppose the arc $r_{j+1} r_j$ is injected by a special pivot
%% at time $t$.  Assume for a contradiction that this arc was defunct at
%% that time, and consider the pivot that caused this arc to become defunct.
%% (That pivot could not have been special.)  The situation looks like
%% THIS.  The pivot is CW, which contradicts a previous lemma.
%% \end{proof}

\bibliographystyle{plain}
\bibliography{refs,other_refs}
\end{document}